\documentclass[conference]{IEEEtran}
\IEEEoverridecommandlockouts

\usepackage{cite}
\usepackage{amsmath,amssymb,amsfonts,amsthm, bbm}
\usepackage{algorithm}
\usepackage{algpseudocode}
\usepackage{graphicx}
\usepackage{subcaption}
\usepackage{textcomp}
\usepackage{xcolor}
\usepackage{tikz}
\usepackage[margin=0.75in]{geometry}
\begin{document}

\newtheorem{problem}{Problem}
\newtheorem{example}{Example}
\newtheorem{definition}{Definition}
\newtheorem{assumption}{Assumption}
\newtheorem{theorem}{Theorem}
\newtheorem{lemma}{Lemma}
\newtheorem{corollary}[theorem]{Corollary}
\newtheorem{proposition}{Proposition}
\newtheorem*{remark}{Remark}
\newtheorem{conjecture}{Conjecture}

\newcommand{\arh}[1]{{\color{blue}{[ARH: #1]}}}
\newcommand{\hmph}[1]{{\color{green}{[Hmph: #1]}}}
\newcommand{\will}[1]{{\color{olive}{[Will: #1]}}}
\newcommand{\ignore}[1]{}

\numberwithin{assumption}{subsection}
\title{Adaptive Identification of SIS Models
    \thanks{*Chi Ho Leung and Philip E. Par\'e are with the Elmore Family School of Electrical and Computer Engineering, Purdue University, USA. William E. Retnaraj and Ashish R. Hota are with the Department of Aerospace Engineering and Electrical Engineering, respectively, IIT Kharagpur, India.  E-mail: leung61@purdue.edu, retnaraj@ieee.org, ahota@ee.iitkgp.ac.in, philpare@purdue.edu. This material is based upon work supported in part by the US-India Collaborative Research Program between the US National Science Foundation (NSF-ECCS \#2032258, \#2238388) and the Department of Science and Technology of India (via IDEAS TIH, ISI Kolkata).}
}
\author{Chi Ho Leung, William E. Retnaraj, Ashish R. Hota, and Philip E. Par\'{e}*}

\maketitle

\begin{abstract}
    Effective containment of spreading processes such as epidemics requires accurate knowledge of several key parameters that govern their dynamics. In this work, we first show that the problem of identifying the underlying parameters of epidemiological spreading processes is often ill-conditioned and lacks the persistence of excitation required for the convergence of adaptive learning schemes. To tackle this challenge, we leverage a relaxed property called initial excitation combined with a recursive least squares algorithm to design an online adaptive identifier to learn the parameters of the susceptible-infected-susceptible (SIS) epidemic model from the knowledge of its states. We prove that the iterates generated by the proposed algorithm minimize an auxiliary weighted least squares cost function. We illustrate the convergence of the error of the estimated epidemic parameters via several numerical case studies and compare it with results obtained using conventional approaches.
\end{abstract}


\section{Introduction}

Accurately forecasting the progression of infectious diseases has become of vital importance following the catastrophic spread of several pandemics over the past century.
Various dynamic models, beginning with the well-known SIS (susceptible-infected-susceptible) model \cite{kermack1927sir} have been developed, often with tailored extensions \cite{giordano2020modelling} in response to different epidemic characteristics.
These models are governed by parameters that dictate different metrics of interest, including the infection trajectory, the peak of infection, and the endemic equilibrium \cite{pastor2015epidemic}.
Therefore, predicting or forecasting the spread of the disease requires estimating these parameters within a limited time in an online manner.


Adaptive identification is an online system identification strategy that updates the parameter estimates based on incoming input and output data \cite{ioannou2006adaptive}. Traditional techniques for addressing adaptive identification problems often involve gradient descent and recursive least square filtering \cite{islam2019recursive}. A longstanding obstacle in successful adaptive estimation and output error convergence is the requirement of a persistently excited regressor signal, which is often overly restrictive \cite{ioannou2006adaptive,narendra1987persistent}. Contemporary solutions to overcome these restrictions involve indirect adaptive control strategies \cite{ioannou1988theory} and robust adaptive control methods \cite{ioannou1996robust}. Other efforts to address the lack of uniform persistent excitation have explored the relaxation of this restrictive requirement, including the concept of initial excitation \cite{jha2019initial, dhar2022initial}.



Another challenge in online system identification is the identifiability of all unknown parameters. While the structural identifiability of a model can be analytically checked in a multitude of ways for both linear and nonlinear systems \cite{hamelin2020observability}, a more pressing requirement is practical identifiability \cite{wieland2021structural, hamelin2020observability}, the lack of which makes parameter estimation a very difficult task despite it being theoretically possible. This challenge has been well-recognized in the context of epidemic models in recent papers \cite{hamelin2020observability,prasse2022predicting}.


In this work we examine the performance of classical adaptive identification algorithms in learning the parameters governing the susceptible-infected-susceptible (SIS) epidemic model. We prove that for this class of models, the regressor matrix is not persistently exciting, which establishes that any algorithm that relies on this condition would not be able to succeed in the parameter identification task. We further show that the regressor matrix is ill-conditioned which leads to practical identifiability issues, and algorithms based on relaxed notions of persistent excitation also do not perform well.

After highlighting the above characteristics, we propose a novel algorithm that builds upon the recursive least squares (RLS) technique, combining the well-known RLS algorithm with the novel concept of an excitation set that is used to construct the main regressor.
We explain the motivation behind this idea, introduce our algorithm, and show that the estimates obtained by the proposed algorithm minimize an auxiliary cost function that weighs exciting data points with unity weight and disregards regressors that are not sufficiently exciting.
We compare our results with existing approaches and demonstrate that our algorithm is effective for adaptive parameter identification of SIS models with noise.

\subsubsection*{Notations}
We denote a matrix $X$ to be positive/negative semi-definite by $X \succeq 0$ and $X \preceq 0$, respectively.
We denote \(\kappa(X)\) as the condition number of matrix \(X\), where the condition number is the ratio of the largest singular value of \(X\) to the smallest singular value of \(X\).
The weighted norm is denoted as $\|x\|_{A} = \sqrt{x^\top Ax}$ for some positive semi-definite matrix $A$, and $\|\cdot\|$ is the 2-norm. 

\section{Problem Formulation}
\label{sec:prob_form}

In this section, we formally present the adaptive identification problem and introduce the necessary notions and tools required to address it.
Subsequently, we justify the need for developing a novel adaptive identification algorithm by providing a motivating example that demonstrates the failure of the classic gradient descent adaptive identification law
in estimating the parameters of the SIS epidemic model.
Consider the class of non-linear discrete-time systems where the parameters are linearly separable from the states:
\begin{equation}\label{eqn:gen_nonlinear_sys}
    x_{k+1} = x_k + \phi(x_k)\theta + \xi_k,
\end{equation}
where $x_k \in \mathcal{X}$ is the state vector, $\theta \in \Theta \subseteq \mathbb{R}_{\geq 0}^p$ is the parameter space, $\phi: \mathcal{X} \mapsto \mathbb{R}^{n \times p}$ is the function that maps states to the regressor matrix, and $\xi_k$ is some bounded unknown perturbation such that $\|\xi_k\| < \nu$ for all $k \geq 0$. We assume that our observation $y_k := x_{k+1} - x_k$ consists of the change in the state variables. We can now write the residuals of each recursive estimation step as:
\begin{equation}\label{eqn:residual}
    r_{k}(\hat{\theta}) = y_k - \phi(x_k)\hat{\theta} = \xi_k.
\end{equation}

Let $\alpha \in (0, 1]$ be the exponential forgetting factor.
Our goal is to design an adaptive identification law $\hat{\theta}_k = \hat{\theta}_{k-1} + f(x_{k-1}, x_{k})$ that minimizes the empirical cost:
\begin{equation}\label{eqn:emp_cost_tseries}
    C_{k}^{(emp)}(\hat{\theta}_k) = \frac{1}{2}\sum_{i=0}^k \alpha^{k-i}\|r_{i}(\hat{\theta}_k)\|^2
\end{equation}
for all $k \geq 0$ given the measurements of $y_k$ and $x_k$.


\subsubsection{Preliminaries}

In the following definitions, we denote $\Psi(\cdot, \cdot, \cdot)$ as the state transition function of the discrete-time dynamics $x_{k+1} = f(x_k)$, such that $x_k = \Psi(k, k_0, x_0)$ when the initial state is $x_0$ at time $k_0$.
Adaptive algorithms are guaranteed to track target values and minimize estimation errors when operating under suitable assumptions, with persistence of excitation being one of the key requirements. There are different variants of definitions of persistent excitation depending on the application\cite{dhar2022initial, jha2019initial, panteley2001relaxed}. We define persistent excitation as follows.
\begin{definition}[Persistent Excitation]
    Let $x_{k+1} = f(x_k)$. A function $\phi(x)$ is said to be persistently exciting w.r.t. $f$ if:
    \begin{equation}\label{eq:pe_definition}
        \sum_{k=l}^{l+L} \phi(\Psi(k, k_0, x_0))^\top \phi(\Psi(k, k_0, x_0)) \succeq \alpha I \quad \forall l \in \mathbb{Z}_{\geq 0}
    \end{equation}
    for some positive constants $L, \alpha$ and initial condition $x_0$.
\end{definition}
Another important aspect of the parameter identification problem is its practical identifiability, which depends on the geometric properties of the Fisher information matrix (FIM) of the cost function defined below.
\begin{definition}[Fisher Information Matrix]\label{def:FIM}
    The Fisher Information Matrix $H$ of the $k^{th}$ empirical cost function \eqref{eqn:emp_cost_tseries} w.r.t. the parameters $\theta$ is defined as $H_{ij} = \partial_{\theta_i}\partial_{\theta_j}C^{(\text{emp})}_k(\theta)$.
\end{definition}
Since the parameters are linearly separable from the observed states, the Fisher Information matrix can be written in the following manner.
\begin{proposition}
    The Fisher Information Matrix of the empirical cost $C^{(\text{emp})}_k(\theta)$ is:
    \begin{equation}
        H = \sum_{i=0}^{k} \alpha^{k-i}\phi(x_i)^\top\phi(x_i).
    \end{equation}
    Further, $H$ is positive semi-definite.
\end{proposition}

The proof is straightforward and is omitted in the interest of space. Another useful notion related to practical identifiability is the condition number \(\kappa\). Notice that $\kappa(M) \in [1, \infty)$ for any real-valued matrix $M$. The condition number of the FIM determines the ill-posedness of the identification problem. If the condition number is infinite, the identification problem is considered to be ill-posed. On the other hand, the problem is considered to be well-posed if the condition number is finite. However, if it is large, the problem may be ill-conditioned, meaning the FIM is close to singular.

\begin{figure}
    \centering
    \begin{subfigure}{\columnwidth}
        \includegraphics[width=\columnwidth]{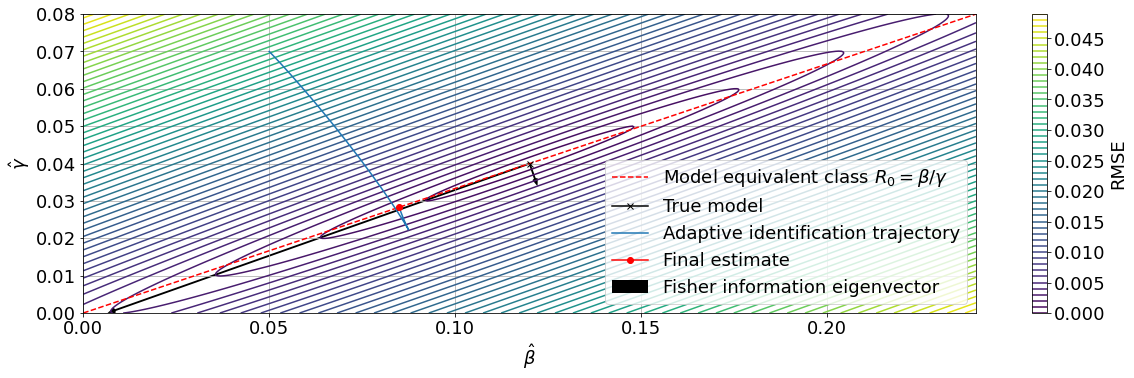}
        \caption{Contour plot of identification error with estimation trajectory}
        \label{fig:ex:contour}
    \end{subfigure}
    \begin{subfigure}{\columnwidth}
        \includegraphics[width=\columnwidth]{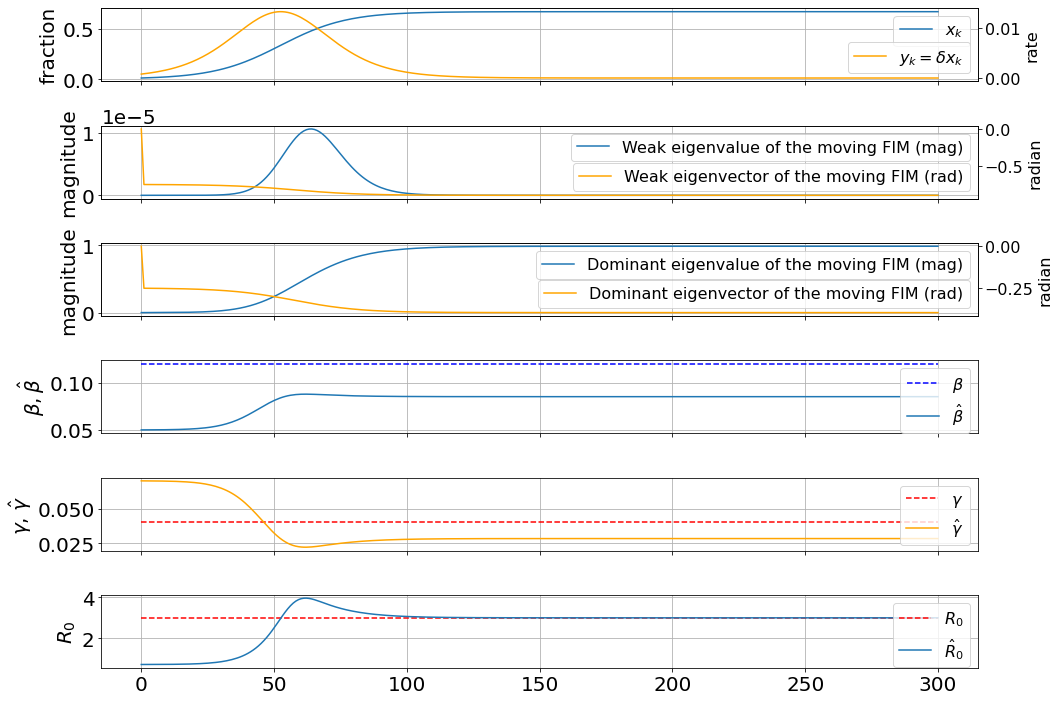}
        \caption{SIS, persistence of excitation, identification error dynamics}
        \label{fig:ex:traj}
    \end{subfigure}
    \caption{
        \footnotesize{
            The trajectory of the estimates on the RMSE contour plot is shown in Fig.~\ref{fig:ex:contour} for the SIS model $(\beta = 0.12$, $\gamma = 0.04$, $x(k_0) = 0.01$, $\hat{\beta}(k_0) = 0.05$, $\hat{\gamma}(k_0) = 0.07)$ without noise added.
            Plots displaying from the top in Fig.~\ref{fig:ex:traj}:
            The infection proportion $x_k$ and net infection rate $\Delta x_k = x_{k+1} - x_k$ over time;
            the next two, the moving FIM's eigenvalues and eigenvector, where the moving FIM is computed as $\sum_{i=l}^{l+3}\phi(x_i)^\top \phi(x_i)$; the next two, system parameters estimates $(\hat{\beta}, \hat{\gamma})$ which fail to converge to the true solution; and lastly, the estimated basic reproduction number computed as $\hat{\mathcal{R}}_0 = \frac{\hat{\beta}}{\hat{\gamma}}$.
        }}
    \label{fig:ex:error_main}
\end{figure}

\subsubsection{Motivating Example}\label{sec:motiv_ex}
In this subsection, we examine the convergence behavior of the classic gradient descent adaptive law applied to the SIS spreading process.
\begin{example}\label{ex:sis_ti_sys}
    Consider the SIS epidemic dynamics given by
    \begin{equation}\label{eq:single_node_SIS}
        x_{k+1} = x_k + (1- x_{k})\beta x_{k} - \gamma x_{k},
    \end{equation}
    where the state $x_k$ represents the proportion of infected individuals in the population, $\beta$ is infection rate, and $\gamma$ is the recovery rate. Consider  the adaptive identification law borrowed from the formalism in \cite{narendra2012stable}:
    \begin{equation}\label{eq:sisExam:AdaIdLaw}
        \hat{\theta}_{k+1} = \hat{\theta}_k + \phi(x_k)^\top\left(y_k - \phi(x_k) \hat{\theta}_k\right),
    \end{equation}
    where $\phi(x_k) := [(1-x_k)x_k\ -x_k]$, $y_k := x_{k+1} - x_{k}$, and $\theta := [\beta \quad \gamma]^\top$.
    Note that $\phi(x_k)^\top(y_k - \phi(x_k) \hat{\theta}_k)$ is known as the negative gradient of the stage cost in \eqref{eqn:emp_cost_tseries}.
    From Fig.~\ref{fig:ex:error_main}, we can see that the classic gradient descent adaptive identification law fails to converge to the true parameters, and the estimates converge to a value in the $\mathcal{R}_0$ equivalent class, which is the set of parameters characterized by having the same basic reproduction number~$\mathcal{R}_0 = \frac{\beta}{\gamma}$. This phenomenon is due to lack of persistent excitation and an ill-conditioned identification problem. 
\end{example}
In the second and third subplots of Fig.~\ref{fig:ex:traj}, note that the FIM's eigenvalues span many orders of magnitude, that is from less than $10^{-5}$ to $1$, over time.
It is worth noting that this phenomenon is not unique to the SIS compartmental model, but exists in nonlinear multi-parameter models from various fields of study \cite{transtrum2015perspective}.
These issues motivate the development of a novel adaptive identification process.

\section{Main Results}


One of the reasons that makes the traditional gradient descent method fail to converge in Example~\ref{ex:sis_ti_sys} is the lack of persistent excitation. The following proposition states that the lack of persistent excitation is a property of the SIS model regardless of the initial condition.

\begin{proposition}
    The regressor matrix $\phi(x_k) := [(1-x_k)x_k\ -x_k]$, $y_k := x_{k+1} - x_{k}$ is not persistently exciting with respect to the SIS dynamics in~\eqref{eq:single_node_SIS}.
\end{proposition}
\begin{proof}
    We first notice that the SIS model in~\eqref{eq:single_node_SIS} has two equilibria: $x^{(1)} = 0$ and $x^{(2)} = 1 - \frac{\gamma}{\beta}$,
    with the epidemiological threshold $\mathcal{R}_0 = \frac{\beta}{\gamma}$.
    When $\mathcal{R}_0 \leq 1$, \eqref{eq:single_node_SIS} is asymptotically stable around $x^{(1)}$; when $\mathcal{R}_0 > 1$, \eqref{eq:single_node_SIS} is asymptotically stable around $x^{(2)}$.
    Fix any $L>0$ and $\alpha>0$; we want to show that there exists an $l$ such that the smallest eigenvalue:
    \begin{equation*}
        \lambda_{\min}\left(\sum_{k=l}^{l+L} \phi(x_k)^\top \phi(x_k)\right) < \alpha,
    \end{equation*}
    where
    \begin{equation*}
        \phi(x_k)^\top \phi(x_k) = x_k^2\begin{bmatrix}
            (1-x_k)^2 & -(1-x_k) \\
            -(1-x_k)  & 1
        \end{bmatrix} =: A(x_k).
    \end{equation*}
    When $\mathcal{R}_0 > 1$, \eqref{eq:single_node_SIS} is asymptotically stable around $x^{(2)}$. Therefore, for all $\epsilon > 0$, there exists a $K$ such that $|x_{k} - x^{(2)}| < \epsilon$ for every $k \geq K$.
    Without loss of generality, assume $x_k$ is approaching $x^{(2)}$ from above, then we can write $x_k = x^{(2)} + \epsilon_k$ where $0 < \epsilon_k < \epsilon$ $\forall~k \geq K$.
    Thus, since $\mathcal{R}_0^{-1} = (1-x^{(2)})$,
    \begin{small}
        \begin{align*}
             & A(x_k) = (x^{(2)} + \epsilon_k)^2\begin{bmatrix}
                                                    (1-x^{(2)} - \epsilon_k)^2 & -(1-x^{(2)} - \epsilon_k) \\
                                                    -(1-x^{(2)} - \epsilon_k)  & 1
                                                \end{bmatrix} \\
             & \quad = (x^{(2)})^2\begin{bmatrix}
                                      \mathcal{R}_0^{-2}  & -\mathcal{R}_0^{-1} \\
                                      -\mathcal{R}_0^{-1} & 1
                                  \end{bmatrix} + \epsilon_k\biggl(
            (x^{(2)})^2\begin{bmatrix}
                           -2\mathcal{R}_0^{-1} + \epsilon_k & 1 \\
                           1                                 & 0
                       \end{bmatrix}                                           \\
             & \ \ \ \ \ \  \qquad +
            (2x^{(2)}+\epsilon_k)\begin{bmatrix}
                                     \left(\mathcal{R}_0^{-1} - \epsilon_k\right)^2 & -(\mathcal{R}_0^{-1} - \epsilon_k) \\
                                     -(\mathcal{R}_0^{-1} - \epsilon_k)             & 1
                                 \end{bmatrix}
            \biggr)                                                                                    \\
             & \quad = (x^{(2)})^2 aa^\top + o(\epsilon)I,
        \end{align*}
    \end{small}
    where $a = \begin{bmatrix}
            \mathcal{R}_0^{-1} & -1
        \end{bmatrix}^\top$.
    We then rewrite $\sum_{k=l}^{l+L}A(x_k)$ as $(L+1)(x^{(2)})^2 aa^\top + o(\epsilon)I$.
    Since $aa^\top$ is rank deficient, we can pick an $\epsilon$ such that $\lambda_{\min}((L+1)(x^{(2)})^2 aa^\top + o(\epsilon)I) < \alpha$.
    When $\mathcal{R}_0 > 1$, we replace $x^{(2)}$ with $0$, which leads to $A(x_k) = o(\epsilon)I$.
    Therefore, the same conclusion is achieved by picking a sufficiently small $\epsilon$ such that $\lambda_{\min}(\sum_{k=l}^{l+L}A(x_k)) = \lambda_{\min}(o(\epsilon)I) < \alpha$.
\end{proof}

The lack of persistent excitation of SIS models is primarily due to the rank deficiency of the FIM induced from $\phi(x_k)$ as $x_k$ approaches an equilibrium. Therefore, we can expect other compartmental models with the number of parameters $p$ greater than the number of informative observed states $n$ to exhibit a similar lack of persistent excitation. In addition, the period of excitation in spreading models usually associates with the short initial transient states of the infection. For this reason, the relaxed notion of persistent excitation introduced in \cite{dhar2022initial, jha2019initial, panteley2001relaxed} could potentially be useful for developing an effective adaptive identification algorithm.


\begin{definition}[Initial Excitation]\label{def:initially_exciting}
    Let $x_{k+1} = f(x_k)$. A function $\phi$ is said to be initially exciting with respect to $f$ if:
    \begin{equation}
        \sum_{k=k_0}^{L} \phi(\Psi(k, k_0, x_0))^\top \phi(\Psi(k, k_0, x_0)) \succeq \alpha I,
    \end{equation}
    for some positive constants $L, \alpha$, and initial condition $x_0$.
\end{definition}

\begin{remark}
    Note that $\sum_{i=0}^{L} \phi(\Psi(k, k_0, x_0))^\top \phi(\Psi(k, k_0, x_0))$ is the Fisher Information matrix of the empirical cost~\eqref{eqn:emp_cost_tseries} at step $L$ when $\alpha=1$.
\end{remark}

However, initial excitation alone is not sufficient for addressing the challenges of adaptive identification for nonlinear spreading processes, illustrated by the following example.

\begin{example}\label{ex:ie_grad_fail}
    In this example, we apply the initial excitation-based multi-model adaptive identification (IE-MMAI) algorithm introduced in \cite{dhar2022initial} to the SIS epidemic with process and observation noise (Figure \ref{fig:ex:mmai}). IE-MMAI fails to converge due to a high FIM condition number, which indicates a highly-skewed error contour induced by the model structure as demonstrated in Fig.~\ref{fig:ex:contour}.
\end{example}

\begin{figure}
    \centering
    \includegraphics[width=\columnwidth, trim = 0 0 0 .28cm, clip]{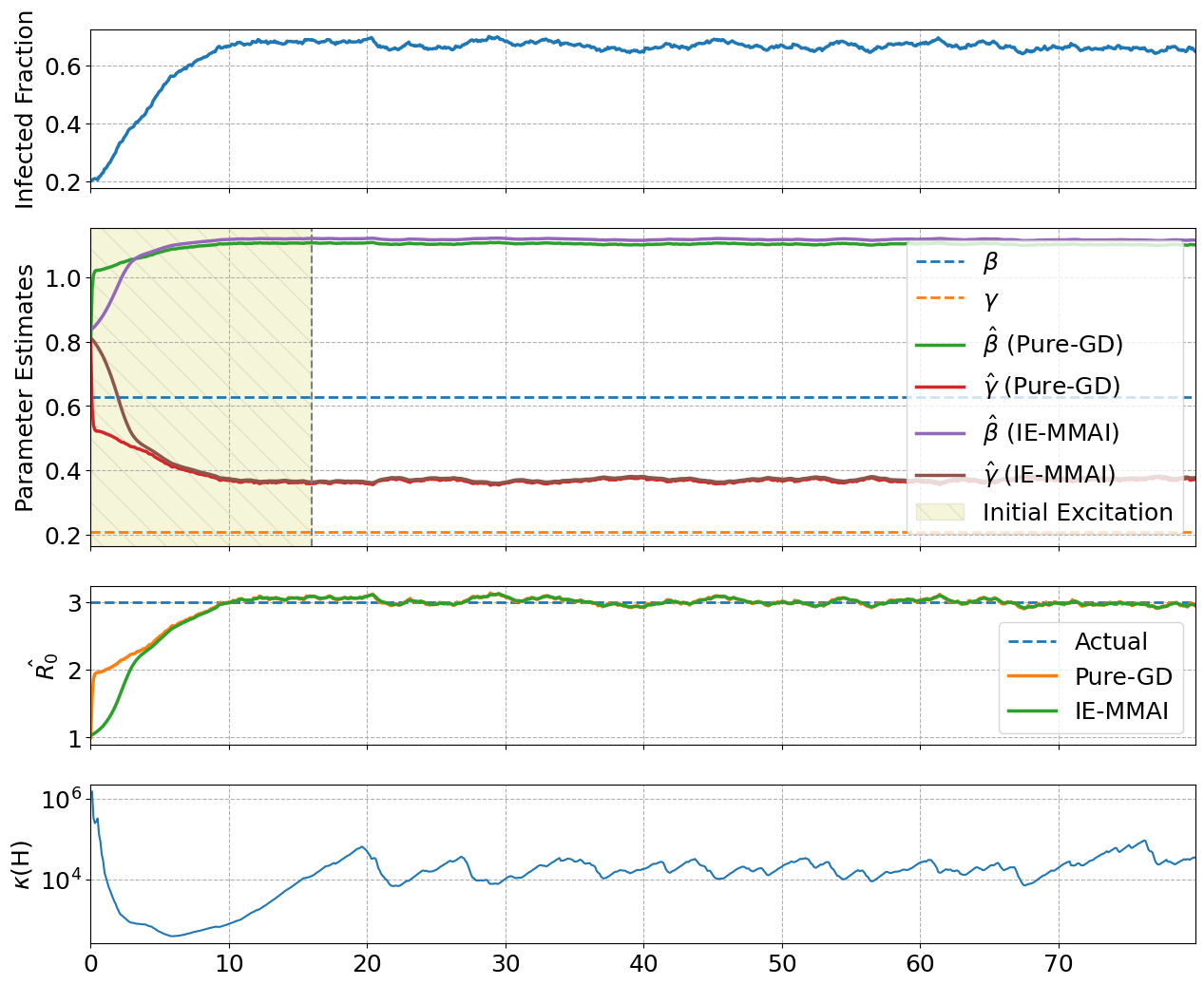}
    \caption{\footnotesize{Plots displaying (from top) the evolution of infected proportion for the SIS model (\(\beta = 0.62929\), \(\gamma = 0.20976\), \(R_0 = 3.0\)) with Gaussian process noise; parameter estimates over time from two chosen baseline algorithms: gradient descent (Pure-GD) (Example \ref{ex:sis_ti_sys}) and initial excitation-based multi-model adaptive identification (IE-MMAI) \cite{dhar2022initial}; and the estimated reproduction number over time using each method; condition number of the Fisher Information Matrix of the empirical cost function (Definition \ref{def:FIM}) in the log scale over time for \(\alpha=0.94\).}}
    \label{fig:ex:mmai}
\end{figure}

Models with the highly-skewed error contour are known to be ``sloppy'' \cite{transtrum2011geometry} or practically unidentifiable \cite{hamelin2020observability} in the literature. Note that, consistent with our discussion in Example~\ref{ex:sis_ti_sys}, the estimated reproduction number converges to the actual value even though the underlying parameters do not.

One natural solution to such a problem is to incorporate second-order information when computing the descent direction upon the arrival of every new data point.
We introduce the following modification to the classic recursive least square algorithm to handle the lack of persistent excitation and practical non-identifiability when applying adaptive identification techniques to the SIS parameter estimation problem.
The proposed algorithm is detailed in Algorithm~\ref{alg:g_rls}, and we refer to it as the Greedily-weighted Recursive Least Squares (GRLS) Algorithm. We now introduce the notion of the \textit{optimal excitation set} and the \textit{greedy excitation set}, inspired by the notion of initial excitation, to illustrate the workings of the GRLS Algorithm.

\begin{definition}[Optimal Excitation Set]\label{def:optimal_set}
    A subset of data points \(\mathcal{E}\) in $\mathcal{D}_K = \{k\in\mathbb{Z}: k_0 \leq k \leq K\}$
    is optimally exciting if:
    \begin{equation*}
        \mathcal{E}_{\text{opt}}(K) = \underset{\substack{E \subseteq \mathcal{D}_K}}{\arg\min}\ \kappa\Big(\sum_{k \in E} \phi(\Psi(k, k_0, x_0))^\top \phi(\Psi(k, k_0, x_0))\Big).
    \end{equation*}
\end{definition}

The main modification of Algorithm~\ref{alg:g_rls} to  classic recursive least square filtering is that it attempts to store the data points that belong to the optimally exciting set so they will not be diluted by less informative new incoming data. However, solving for the optimally exciting set at each update iteration is expensive. Therefore, we propose a feasible approach to obtain a sub-optimal excitation set via a greedy algorithm.

\begin{definition}[Greedy Excitation Set]\label{def:greedyset}
    The $K^{th}$ data point belongs to the greedy excitation set $\mathcal{E}_{\text{g}}(K)$ if it does not deteriorate the FIM's condition number induced by $\mathcal{E}_{\text{g}}(K-1)$:
    \begin{small}
        \begin{equation*}
            \kappa\Big(\sum_{\substack{ k \in\mathcal{E}_{\text{g}}(K-1)\\ \bigcup \{K\}}} \phi(x_k)^\top \phi(x_k)\Big) \leq \kappa\Big(\sum_{k \in \mathcal{E}_{\text{g}}(K-1)} \phi(x_k)^\top \phi(x_k)\Big).
        \end{equation*}
    \end{small}
\end{definition}

\begin{algorithm}
    \caption{Greedily-weighted Recursive Least Squares}\label{alg:g_rls}
    \begin{algorithmic}[1]
        \Require Input datum $x_{k+1}$
        \If{$k < 0$}
        Initialize $H^{(e)}_0$, $P_0$, $\hat{\theta}_0$, $\Phi^{(e)}_0$, $\upsilon^{(e)}_0$.
        \Else
        \State $H^{(e)} \gets H^{(e)}_k + \phi(x_k)^\top\phi(x_k)$\label{alg:g_rls:line:greedy_begin}
        \If{$\kappa(H^{(e)}) \leq \kappa{(H^{(e)}_k)}$}\label{alg:g_rls:line:best_condition_num_sf}
        \State $\Phi^{(e)}_{k+1}, H^{(e)}_{k+1} \gets [{\Phi^{(e)}_k}^\top, \phi(x_k)^\top]^\top, H^{(e)}$\label{alg:g_rls:line:define_Phi_e_in_greedy}
        \State $\upsilon^{(e)}_{k+1} \gets \upsilon^{(e)}_k + \phi(x_k)^\top(x_{k+1} - x_k)$\label{alg:g_rls:line:define_upsilon_e_in_greedy}
        \State $\Phi \gets \sqrt{1-\alpha}\Phi^{(e)}_{k+1}$\label{alg:g_rls:line:define_Phi_in_greedy}
        \State $H, \upsilon \gets (1-\alpha)H^{(e)}_{k+1}, (1-\alpha)\upsilon^{(e)}_{k+1}$\label{alg:g_rls:line:define_var_in_greedy}
        \Else
        \State $\Phi^{(e)}_{k+1}, H^{(e)}_{k+1}, \upsilon^{(e)}_{k+1} \gets \Phi^{(e)}_{k}, H^{(e)}_k, \upsilon^{(e)}_{k}$\label{alg:g_rls:line:define_var_out_greedy}
        \State $\Phi \gets [\sqrt{1-\alpha}{\Phi^{(e)}_{k+1}}^\top, \phi(x_k)^\top]^\top$\label{alg:g_rls:line:define_Phi_out_greedy}
        \State $H \gets (1-\alpha)H^{(e)}_{k+1} + \phi(x_k)^\top\phi(x_k)$\label{alg:g_rls:line:define_H_out_greedy}
        \State $\upsilon \gets (1-\alpha)\upsilon^{(e)}_{k+1} + \phi(x_k)^\top(x_{k+1} - x_k)$\label{alg:g_rls:line:define_upsilon_out_greedy}
        \EndIf\label{alg:g_rls:line:greedy_end}
        \State $P_{k+1} \gets \frac{1}{\alpha}P_k - \frac{1}{\alpha}P_k{\Phi}^\top (\alpha I + \Phi P_k \Phi^\top)^{-1} \Phi P_k$\label{alg:g_rls:line:P_update}
        \State $\hat{\theta}_{k+1} \gets \hat{\theta}_{k} + P_{k+1}(\upsilon - H\hat{\theta}_k)$\label{alg:g_rls:line:theta_update}
        \EndIf
    \end{algorithmic}
\end{algorithm}
After the introduction of the greedy excitation set, we can briefly summarize the GRLS Algorithm as follows.
\begin{itemize}
    \item Lines~\ref{alg:g_rls:line:greedy_begin}-\ref{alg:g_rls:line:greedy_end} in Algorithm~\ref{alg:g_rls} update the greedy exciting set and the corresponding hyper-parameters for computing $\hat{\theta}_{k+1}$ upon every incoming datum. \begin{itemize}
              \item Line~\ref{alg:g_rls:line:best_condition_num_sf} compares the new condition number with the current condition number.
              \item If $\kappa(H^{(e)}) \leq \kappa{(H^{(e)}_k)}$, then we update the regressor matrices $\Phi^{(e)}_{k+1}$, the FIM $H^{(e)}_{k+1}$ and the corrector term $\upsilon^{(e)}_{k+1}$ of the excitation set on Lines~\ref{alg:g_rls:line:define_Phi_e_in_greedy}-\ref{alg:g_rls:line:define_upsilon_e_in_greedy};
              \item else, $\Phi^{(e)}_{k+1}, H^{(e)}_{k+1}, \upsilon^{(e)}_{k+1}$ remain unchanged as specified in Line~10.
              \item Lines \ref{alg:g_rls:line:define_Phi_in_greedy}-\ref{alg:g_rls:line:define_var_in_greedy}, \ref{alg:g_rls:line:define_Phi_out_greedy}-\ref{alg:g_rls:line:define_upsilon_out_greedy} use $\Phi^{(e)}_{k+1}, H^{(e)}_{k+1}, \upsilon^{(e)}_{k+1}$ to compute $\Phi, H, \upsilon$ which are needed for computing the inverse Hessian matrix $P_{k+1}$ and the new estimate $\hat{\theta}_{k+1}$.
          \end{itemize}
    \item Line~\ref{alg:g_rls:line:P_update} updates the inverse Hessian matrix with an exponential forgetting factor $\alpha$.
    \item Line~\ref{alg:g_rls:line:theta_update}
          updates the parameter estimates.
\end{itemize}
\noindent
We are ready to characterize the optimality of Algorithm~\ref{alg:g_rls}.
\begin{theorem}\label{thm:opt_theta}
    If $P_0$ is positive definite,
    then for all $k \in [0, \infty)$, $\hat{\theta}_{k+1}$ obtained by Algorithm~\ref{alg:g_rls} is the unique minimizer of the cost function:
    \begin{align}\label{eq:cost}
        C_{k}(\hat{\theta}_k) = \sum_{i=0}^{k}w_{i,k}\|r_{i}(\hat{\theta}_k)\|^2 + \alpha^{k+1}\|\hat{\theta}_k - \theta_0\|_{P_0^{-1}}^2
    \end{align}
    with the weighting function defined as:
    \begin{equation*}
        w_{i,k} =
        \begin{cases}
            (1-\alpha)\sum_{l=i}^{k}\alpha^{k-l} & \text{if $i \in \mathcal{E}_{\text{g}}(k)$,} \\
            \alpha^{k-i}                         & \text{otherwise}.
        \end{cases}
    \end{equation*}
\end{theorem}
\begin{proof}
    The proof leverages mathematical induction to proceed, and it suffices to show the inductive step.
    We first note that $C_{k}(\hat{\theta})$ can be written in terms of:
    $C_{k}(\hat{\theta}) = \hat{\theta}^\top A_k\hat{\theta} + 2b_k^\top \hat{\theta} + c_k$,
    where $A_k, b_k, c_k$ are:
    \begin{align*}
        A_k & := \sum_{i=0}^k w_{i, k} \phi(x_i)^\top \phi(x_i) + \alpha^{k+1}P^{-1}_0                  \\
        b_k & := -\sum_{i=0}^k w_{i, k} \phi(x_i)^\top y(x_i) - \alpha^{k+1}P^{-1}_0\theta_0            \\
        c_k & := \sum_{i=0}^k w_{i, k} y(x_i)^\top y(x_i) + \alpha^{k+1}\theta_0^\top P^{-1}_0\theta_0.
    \end{align*}
    Then, $A_k, b_k$ can be computed recursively as
    \begin{align*}
        A_k & = \alpha A_{k-1} +\sum_{i\in U}\omega_i \phi(x_i)^\top\phi(x_i) \\
        b_k & = \alpha b_{k-1} - \sum_{i\in U}\omega_i \phi(x_i)^\top y_i,
    \end{align*}
    where $U = \mathcal{E}_{\text{g}}(k) \bigcup \{k\}$, and  
    $
        \omega_i = \begin{cases}
            (1 - \alpha) & \text{if $i \in \mathcal{E}_{\text{g}}(k)$}, \\
            1            & \text{otherwise}.
        \end{cases}
    $
    By way of induction, assume $\exists~k\in \mathbb{N}$ such that $A_{k-1}$ is positive definite and the unique optimizer of $C_{k-1}(\hat{\theta}_{k-1})$ is $\hat{\theta}_{k} = -A_{k-1}^{-1}b_{k-1}$.
    We define $P_{k+1} := A_k^{-1}$.
    Since $A_{k-1}$ is positive definite, 
    we can apply the matrix inversion lemma\cite[p, 304]{bernstein2009matrix} and obtain a positive definite $P_{k+1}$:
    \begin{align*}
        P_{k+1} & = A_k^{-1}                                                                                                              \\
                & = \frac{1}{\alpha}\left(A_{k-1} +\frac{1}{\alpha}\left(\sum_{i\in U}\omega_i \phi(x_i)^\top\phi(x_i)\right)\right)^{-1} \\
                & = \frac{1}{\alpha}P_k - \frac{1}{\alpha}P_k{\Phi}^\top (\alpha I + \Phi P_k \Phi^\top)^{-1} \Phi P_k,
    \end{align*}
    where
    $
        \Phi = \begin{cases}
            \sqrt{1 - \alpha}\Phi^{(e)}_{k+1}                               & \text{if $k \in \mathcal{E}_{\text{g}}(k)$,} \\
            [\sqrt{1 - \alpha}(\Phi^{(e)}_{k+1})^\top, \phi(x_k)^\top]^\top & \text{otherwise},
        \end{cases}
    $
    and $\Phi^{(e)}_{k+1} = \begin{bmatrix}
            \phi(x_{k_1})^\top, \phi(x_{k_2})^\top, \dots, \phi(x_{k_n})^\top
        \end{bmatrix}_{k_i\in \mathcal{E}_{\text{g}}(k)}^\top$, which satisfies the computation of $\Phi^{(e)}_{k+1}$ and $\Phi$ on Lines~\ref{alg:g_rls:line:define_Phi_e_in_greedy}, \ref{alg:g_rls:line:define_Phi_in_greedy}, \ref{alg:g_rls:line:define_var_out_greedy}, and \ref{alg:g_rls:line:define_Phi_out_greedy} of Algorithm~\ref{alg:g_rls}.
    By the quadratic minimization lemma~\cite{bernstein2009matrix}, the unique minimizer of $C_{k}(\hat{\theta})$ is:
    \begin{small}
        \begin{align*}
            \hat{\theta}_{k+1} & = -A_k^{-1}b_k                                                                                                                                         \\
                               & = A_k^{-1}\Big(-\alpha b_{k-1} + \sum_{i \in U}\omega_i\phi(x_i)^\top y_i\Big)                                                                         \\
                               & = A_k^{-1}\Big(\alpha A_{k-1} \hat{\theta}_k + \sum_{i \in U}\omega_i\phi(x_i)^\top y_i\Big)                                                           \\
                               & = A_k^{-1}\Big(\Big(A_k - \sum_{i \in U}\omega_i\phi(x_i)^\top\phi(x_i)\Big)\hat{\theta}_k + \sum_{i \in U}\omega_i\phi(x_i)^\top y_i\Big)             \\
                               & = \hat{\theta}_k + A_k^{-1}\Big(\sum_{i \in U}\omega_i\phi(x_i)^\top y_i - \Big(\sum_{i \in U}\omega_i\phi(x_i)^\top\phi(x_i)\Big)\hat{\theta}_k \Big) \\
                               & = \hat{\theta}_k + P_{k+1}(\upsilon - H\hat{\theta}_k),
        \end{align*}
    \end{small}

    \noindent
    where $\upsilon := \sum_{i \in U}\omega_i\phi(x_i)^\top y_i$ and $H := \sum_{i \in U}\omega_i\phi(x_i)^\top\phi(x_i)$ matches the computation of $\upsilon$ and $H$ on Lines~\ref{alg:g_rls:line:define_var_in_greedy},
    \ref{alg:g_rls:line:define_H_out_greedy}, and
    \ref{alg:g_rls:line:define_upsilon_out_greedy}
    of Algorithm~\ref{alg:g_rls}.
    By the principle of mathematical induction, $P_n$ is positive definite and $\hat{\theta}_{n+1}$ is the unique minimizer of $C_n(\hat{\theta}_{n})$ for all $n \in \mathbb{N}$.
\end{proof}

While Theorem~\ref{thm:opt_theta} characterizes the cost function which Algorithm~\ref{alg:g_rls} optimizes, the structure of the weighting function $w_{i,k}$ might not be immediately obvious.
The following corollary clarifies the intuition behind $w_{i,k}$.
\begin{corollary}\label{cor:limit_cost}
    If $\alpha < 1$, $P_0$ is positive definite, and the cardinality of the excitation set in the limit, $\lim_{k\to \infty}|\mathcal{E}_g(k)|$, is finite,
    then, as $k \to \infty$, $\hat{\theta}_{k+1}$ obtained by Algorithm~\ref{alg:g_rls} is the unique minimizer of the cost function: 
    \begin{align}\label{eq:cost2}
        C_{k}(\hat{\theta}) = \sum_{i=0}^{k}W_{i,k}\|r_i(\hat{\theta})\|^2,
    \end{align}
    with
    \begin{equation}\label{eq:cor:limit_weight}
        W_{i,k} = \begin{cases}
            1            & \text{if $i \in
            \mathcal{E}_g(K)$}               \\
            \alpha^{k-i} & \text{otherwise}.
        \end{cases}
    \end{equation}
\end{corollary}
\begin{proof}
    Note that, since $\alpha<1$, the second term in the cost function in \eqref{eq:cost}, $\alpha^{k+1}\|\hat{\theta} - \theta_0\|_{P_0^{-1}}^2$, goes to zero as $k \to \infty$.
    Furthermore, $\sum_{l=i}^{k}\alpha^{k-l}$ can be rewritten as
    $\sum_{l=0}^{k-i}\alpha^{l}$,
    and, since $\sum_{l=0}^{k-i}\alpha^{l}$ is a geometric sum, $\alpha < 1$, and $i\in \mathcal{E}_g(k)$ is finite by assumption,
    $\lim_{k\to \infty} (1-\alpha)\sum_{l=i}^{k}\alpha^{k-l} = 1$. 
    Thus, the weighting function $w_{i,k}$ can be written as $W_{i,k}$ in \eqref{eq:cor:limit_weight}.  
    Therefore, by Theorem~\ref{thm:opt_theta}, Algorithm~\ref{alg:g_rls} obtains the optimal $\theta^* = {\arg\min}_{\hat{\theta}} C_{k}(\hat{\theta})$ for the 
    cost function in \eqref{eq:cost2}-\eqref{eq:cor:limit_weight}.
\end{proof}
Corollary~\ref{cor:limit_cost} characterizes the asymptotic behavior of the cost function which Algorithm~\ref{alg:g_rls} minimizes.
The algorithm stores the data points which belong to the greedy excitation set by giving them a weight of $1$ while assigning the rest of the unexciting points exponentially decaying weights.

\section{Simulations}

We present parameter estimation results for both the noise-free case and including Gaussian process and observation noise for the SIS dynamics in~\eqref{eq:single_node_SIS}. In Fig.~\ref{fig:compare_scalar_SIS}, we compare the performance of the initial excitation approach, IE-MMAI \cite{dhar2022initial}, and a basic RLS with an exponential forgetting (EF-RLS), to the performance of the GRLS Algorithm we propose in this work.
The same initial estimates of parameters, $\theta_0 = [\beta_0, \gamma_0]^\top = [1, 1]^\top$,
are used for all algorithms with the exception of IE-MMAI, for which the \(m=3\) models were initialized randomly around
\(\theta_0\).
\begin{figure}[h]
    \centering
    \begin{subfigure}[t]{0.49\columnwidth}
        \includegraphics[width=\columnwidth, trim = 0 0cm 0 .1cm, clip]{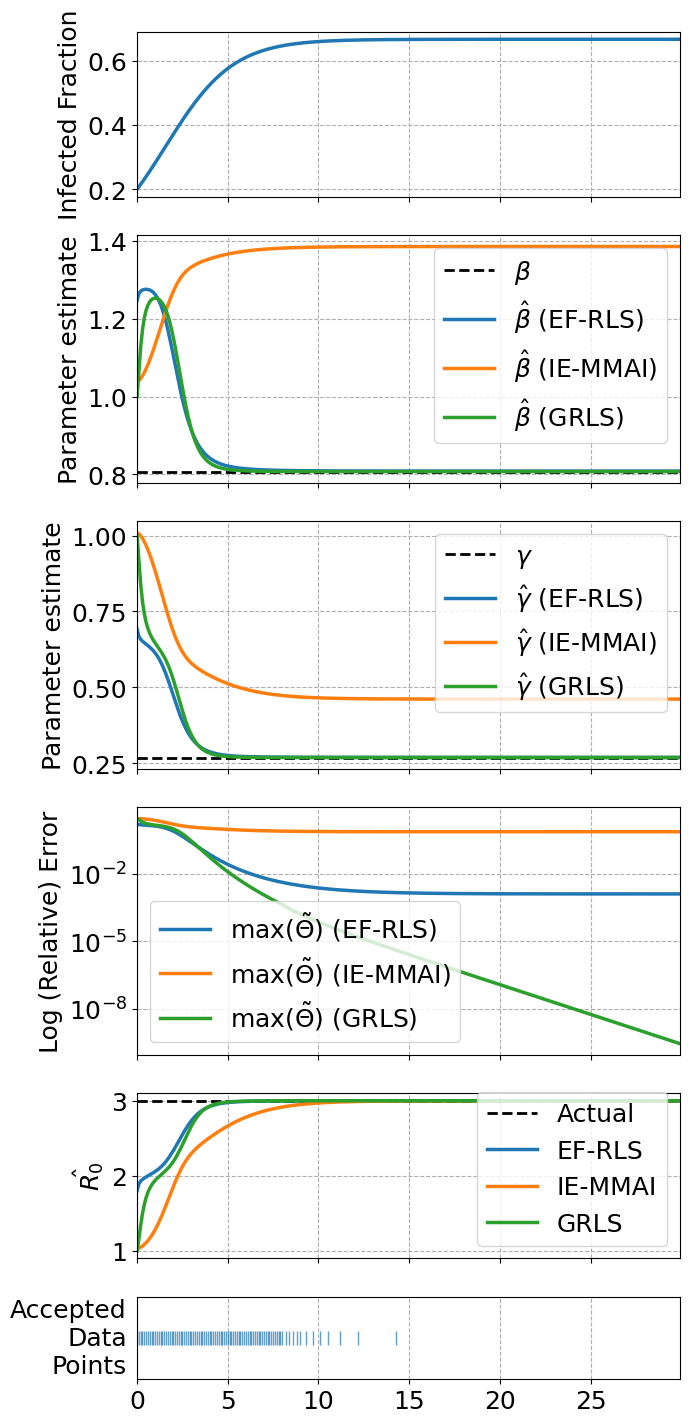}
    \end{subfigure}
    \begin{subfigure}[t]{0.49\columnwidth}
        \includegraphics[width=\columnwidth, trim = 0 0cm 0 .1cm, clip]{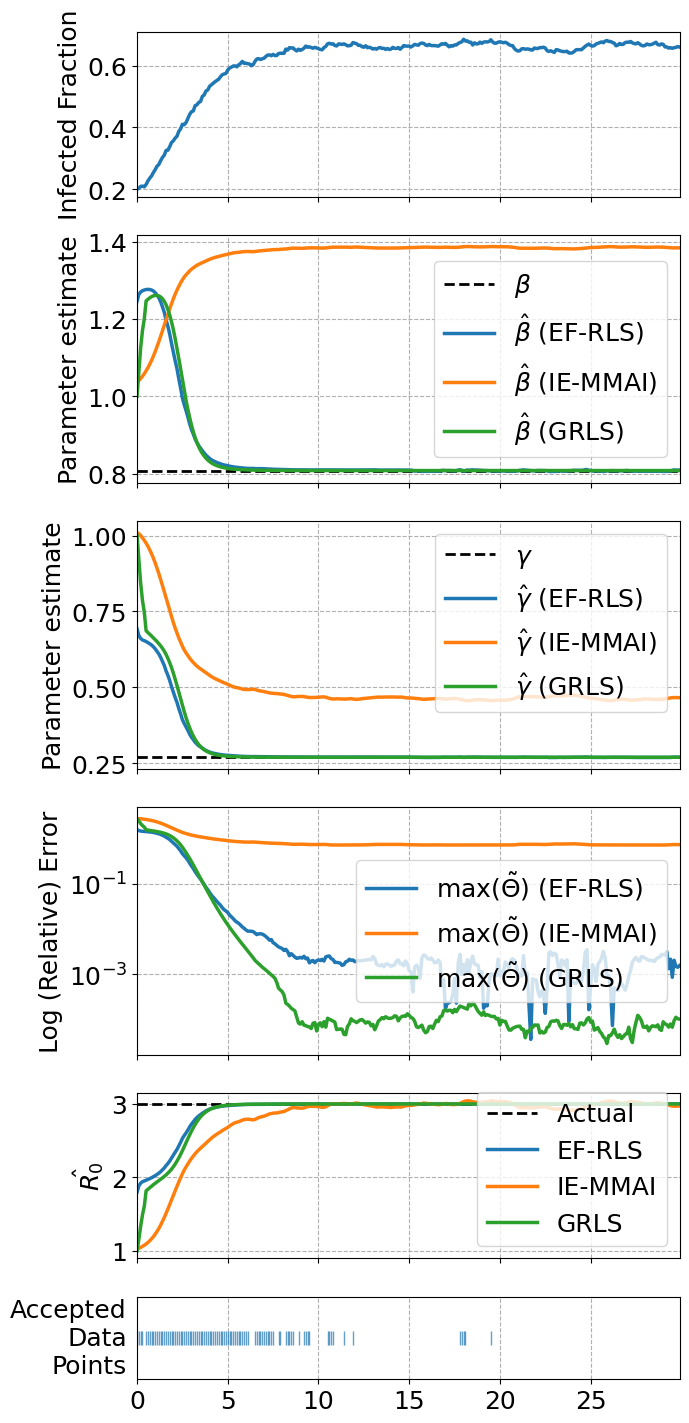}
    \end{subfigure}
    \caption{\footnotesize{Comparing the performance of EF-RLS, IE-MMAI, and GRLS: SIS Simulation (\(\beta = 0.8076\), \(\gamma = 0.2692\));
            parameter estimates over time for both a noise-free and a noisy simulation;
            maximum relative error in log scale for parameter estimates over time;
            reproduction number estimates over time;
            and data indices accepted into $\mathcal{E}_g$ by GRLS.
            Left panel shows the results without noise while the right panel shows it in presence of noise.}}
    \label{fig:compare_scalar_SIS}
\end{figure}
IE-MMAI, designed for only LTI systems and as a gradient descent-like first-order method, is justifiably sensitive to the choice of the initial parameter estimates, and often fails to converge due to the poor practical identifiability of SIS models, as discussed in Examples~\ref{ex:sis_ti_sys} and \ref{ex:ie_grad_fail}.
Note that the estimated reproduction numbers still converge to the actual value despite the lack of convergence of the parameter estimates themselves, which is consistent with our discussion of SIS practical identifiability in Section \ref{sec:motiv_ex}.
On the other hand, even in the presence of noise, the parameter estimates converge for both  EF-RLS and GRLS.

Data accepted into the GRLS greedy excitation set (\ref{def:greedyset}) used to construct the main regressor are diagrammatically depicted in the lowermost block of Fig.~\ref{fig:compare_scalar_SIS}. A majority of points accepted by the algorithm is in the transient rise of states before the equilibrium is reached. The beginning of the epidemic garners a critical amount of information about the epidemic parameters.

Note that while EF-RLS is comparable in performance to GRLS for the noise-free case (left panel of Fig.~\ref{fig:compare_scalar_SIS}), it becomes increasingly oscillatory upon losing excitation in the noisy case (right panel of Fig.~\ref{fig:compare_scalar_SIS}). A closer look at the covariance matrix \(P_k\) in both algorithms reveals a steady increase in the condition number and maximum eigenvalue of \(P_k\) in EF-RLS, while those of GRLS saturate due to its tendency to avoid picking up non-exciting data points
(see Fig.~\ref{fig:Pk_RLS}).
This phenomenon in EF-RLS is observed in both the noise-free and noisy cases and is known in literature as covariance windup \cite{fortescue1981selftune}, which occurs when a non-unity forgetting factor in EF-RLS causes \(P_k\) to get closer to a singular matrix upon losing persistence of excitation. The linear increase in the maximum eigenvalue of \(P_k\) is consistent with past analysis~\cite{cao2001windup}.
Upon running the parameter estimation task for longer times, the EF-RLS estimates diverge.

\begin{figure}[h]
    \centering
    \includegraphics[width=0.9\columnwidth]{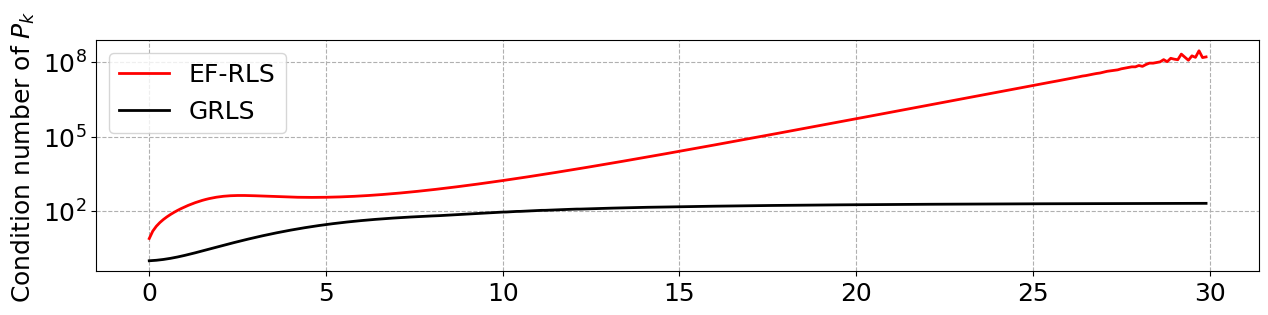}
    \caption{\footnotesize{Plot of the condition number of the covariance matrix \(P_k\) against time for both EF-RLS and GRLS running on data from the noise-free SIS simulation (Fig.~\ref{fig:compare_scalar_SIS}).}}
    \label{fig:Pk_RLS}
\end{figure}

While there are other ways to combat covariance windup in RLS and Kalman filtering applications \cite{lai2023er_rls, cao2001windup}, our approach seeks to mitigate the issues uncovered by Example~\ref{ex:ie_grad_fail} as a first step to discover even more effective methods for online parameter estimation.

\ignore{
\subsection{Networked SIS Model}
The same methodology may also be extended to a networked SIS system with an appropriate modification to the regressor \(\phi(x)\) in \eqref{eqn:gen_nonlinear_sys} as follows:
\begin{equation*}
    \phi(x(t)) =
    \begin{bmatrix}
        x(t)^\top \otimes \text{diag}(1 - x(t)) & - \text{diag}(x(t))
    \end{bmatrix},
\end{equation*}
where \(x(t)\) is an \(n\)-dimensional vector and \(n\) is the number of nodes in the network in consideration. The parameter vector \(\theta\) is the concatenated flattened adjacency matrix \(\text{vec}(B)\) (of infection rates between different nodes, not assumed symmetric) and recovery rates vector \(g\), making it a (\(n^2 + n\))-dimensional vector. We show results for the application of IE-MMAI to data generated by two networked SIS simulations (\(n=7\)) generated by a Star topology graph and an Erd\H{o}s-R\'{e}nyi (ER) graph with the probability of edge generation as \(0.80\). The adjacency matrices were generated with the help of python package \texttt{networkx}. Results are presented in the form of selected parameter estimates and the worst log relative error over time in Fig.~\ref{fig:NetworkSIS}.

\begin{figure}
    \centering
    \begin{subfigure}[t]{0.48\columnwidth}
        \includegraphics[width=\columnwidth]{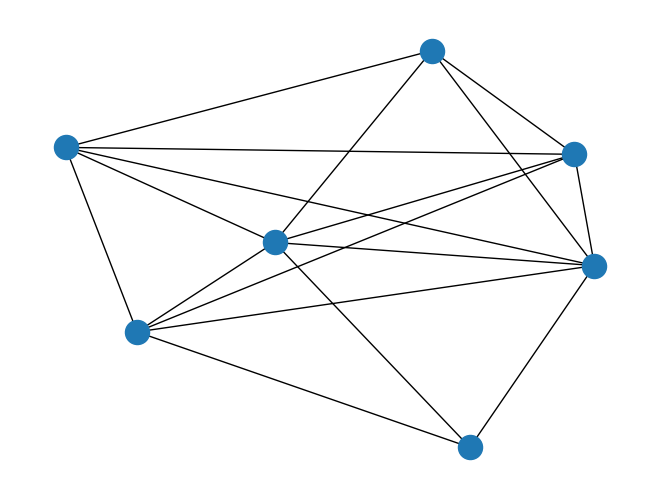}
        \caption{Erd\H{o}s-R\'{e}nyi Network}
        \label{fig:res:Net_ER}
        \vspace*{2mm}
    \end{subfigure}
    \begin{subfigure}[t]{0.48\columnwidth}
        \includegraphics[width=\columnwidth]{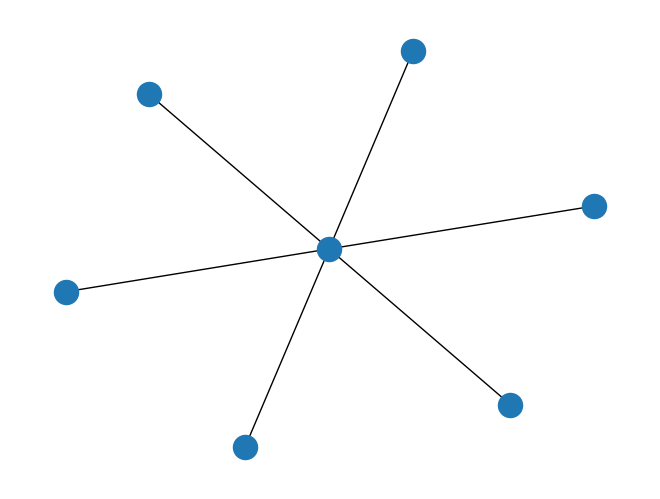}
        \caption{Star Network}
        \label{fig:res:Net_Star}
        \vspace*{2mm}
    \end{subfigure}
    \begin{subfigure}[t]{0.48\columnwidth}
        \includegraphics[width=\columnwidth]{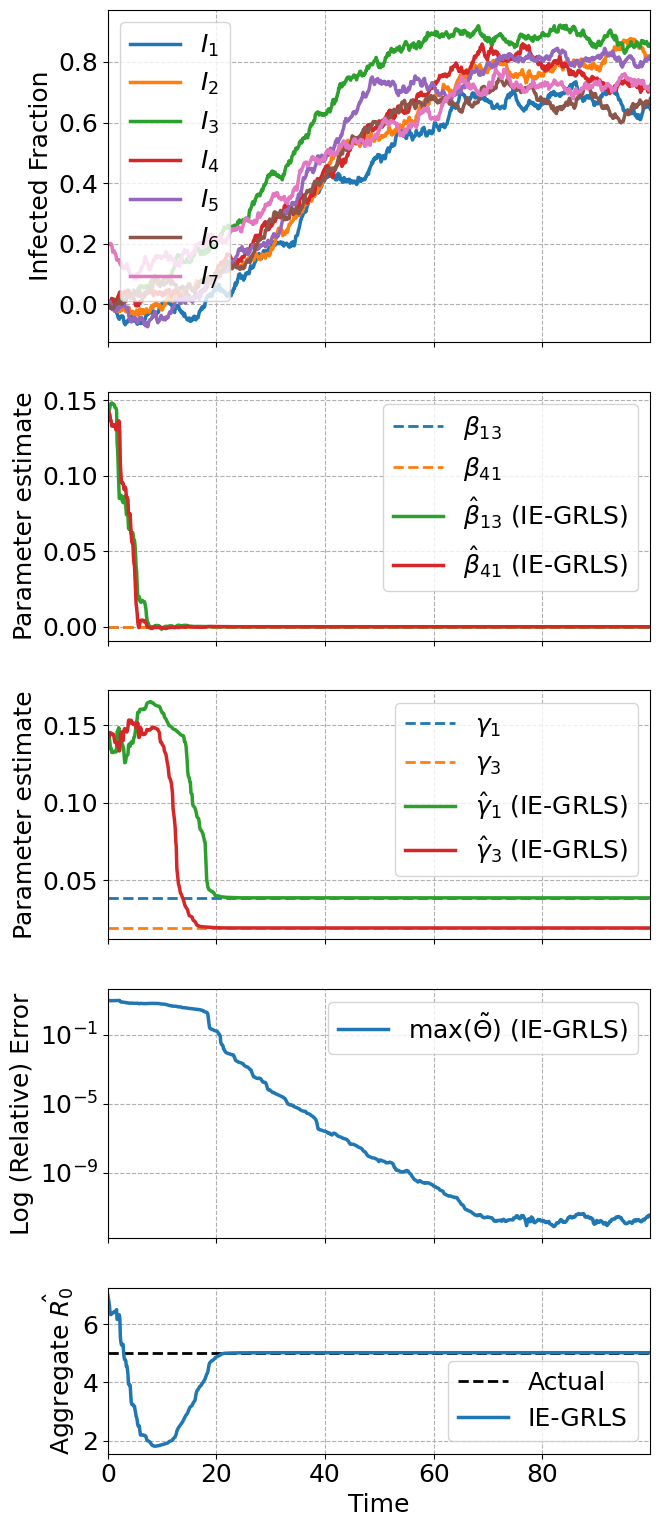}
        \caption{Erd\H{o}s-R\'{e}nyi Network Results}
        \label{fig:res:Net_ER_res}
        \vspace*{2mm}
    \end{subfigure}
    \begin{subfigure}[t]{0.488\columnwidth}
        \includegraphics[width=\columnwidth]{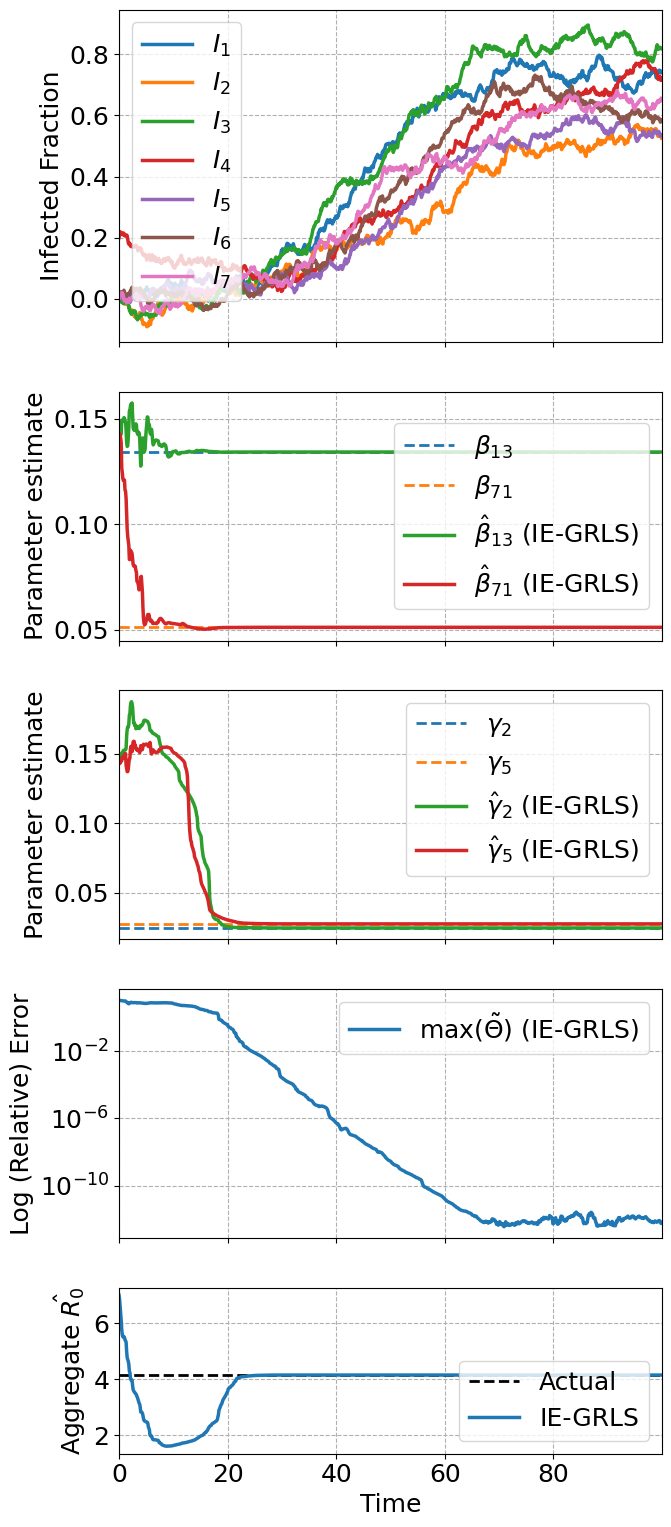}
        \caption{Star Network Results}
        \label{fig:res:Net_Star_res}
        \vspace*{2mm}
    \end{subfigure}
    \caption{\footnotesize{GRLS is applied to a Networked SIS (\(n=7\)) with process noise for two topologies shown in (\subref{fig:res:Net_ER}) and (\subref{fig:res:Net_Star}), the results of which are shown in (\subref{fig:res:Net_ER_res}) and (\subref{fig:res:Net_Star_res}), (from top): Infected fraction in all nodes over time; Selected parameter estimates from the \(B\) matrix; and likewise from the \(g\) vector; worst log relative error in parameter estimates over time, and the estimated aggregate network reproduction number over time.}}
    \label{fig:NetworkSIS}
\end{figure}
}

\section{Conclusion}
We have highlighted two problems that plague the application of adaptive identification tools to SIS models: the lack of persistence of excitation and the practical non-identifiability of epidemic models.
We propose a novel algorithm (GRLS) based on recursive least squares and use the concept of initial excitation to construct an exciting set for the regressor.
The GRLS Algorithm has superior performance compared to conventional algorithms and is able to identify epidemic parameters in the SIS model with process and observation noise.
In particular, while estimates from EF-RLS become oscillatory in the presence of noise, eventually diverging, GRLS is able to maintain a stable estimate.
Future work includes extending GRLS to estimate epidemic parameters on networks, time-varying systems, and other epidemic compartmental models, with the eventual goal of performing adaptive identification using real testing data.


\bibliographystyle{IEEEtran}

\end{document}